\documentclass[11pt,letterpaper]{article}
\pdfoutput=1
\usepackage[final=true,colorlinks=true,citecolor=blue]{hyperref}

\usepackage[utf8]{inputenc}
\usepackage{algorithm}
\usepackage{algorithmicx}
\usepackage[noend]{algpseudocode}
\usepackage[noadjust]{cite}
\usepackage[OT4]{fontenc}
\usepackage{amsthm}
\usepackage{amssymb}
\usepackage{amsmath}
\usepackage{amsfonts}
\usepackage{todonotes}
\usepackage{rotating}
\usepackage{xspace}
\usepackage[shortlabels]{enumitem}
\usepackage{verbatim}
\usepackage{mathrsfs}
\usepackage{framed}
\usepackage{thm-restate}
\usepackage[capitalise]{cleveref}
\usepackage{colortbl}
\usepackage{array,multirow}

\def\dd{\mathinner{.\,.}}
\newcommand{\cO}{\mathcal{O}}
\newcommand{\Oh}{\cO}
\newcommand{\cOtilde}{\tilde{\cO}}
\newcommand{\Ohtilde}{\cOtilde}
\newcommand{\eps}{\epsilon}
\newcommand{\stringeps}{\varepsilon}
\newcommand{\X}{\widetilde{X}}

\newcommand{\T}{\widetilde{T}}
\renewcommand{\L}{\mathcal{L}}
\newcommand{\dH}{\delta_{\operatorname{H}}}

\newcommand{\Mis}{\operatorname{Mis}}
\newcommand{\Occ}{\operatorname{Occ}_k^{\operatorname{H}}}
\newcommand{\EDSMmis}{\textsc{EDSM with $k$ Mismatches}\xspace}
\newcommand{\EDSMerr}{\textsc{EDSM with $k$ Errors}\xspace}

\newcommand{\APEmis}{\textsc{APE with $k$ Mismatches}\xspace}
\DeclareMathOperator{\poly}{poly}
\DeclareMathOperator{\polylog}{polylog}
\newcommand{\Ext}[1][k]{\operatorname{Ext}_{#1}^{\operatorname{H}}}
\newcommand{\eq}[1]{\begin{align*} #1 \end{align*}}
\newcommand{\set}[1]{\left\lbrace #1 \right\rbrace}
\newcommand{\QP}{\bar{P}}
\newcommand{\QT}{\bar{T}}

\usepackage[left=1in,right=1in,top=1in,bottom=1in]{geometry}
\usepackage{authblk}

\newtheorem{theorem}{Theorem}[section]

\newtheorem{lemma}{Lemma}

\makeatletter
\newcommand\thankssymb[1]{\textsuperscript{\@fnsymbol{#1}}}
\makeatother

\title{Faster ED-String Matching with $k$ Mismatches}
\date{}

\author[1]{Paweł Gawrychowski\thanks{Partially supported by the Polish National
Science Centre grant number 2023/51/B/ST6/01505.}}
\author[1]{Adam Górkiewicz\thankssymb{1}}
\author[1]{Pola Marciniak\thankssymb{1}}
\author[2]{Solon P. Pissis}
\author[1]{Karol Pokorski}

\affil[1]{Institute of Computer Science, University of Wrocław, Poland}
\affil[2]{CWI, Amsterdam, The Netherlands}

\newcommand{\genproblem}[3]{
\begin{framed}
  \noindent
  \textbf{Problem:} \textsc{#1}

  \noindent
  \textbf{Input:} #2

  \noindent
  \textbf{Output:} #3
\end{framed}
}

\begin{document}
\maketitle

\thispagestyle{empty}

\begin{abstract}
We revisit the complexity of approximate pattern matching in an elastic-degenerate string. Such a string is a sequence
of $n$ finite sets of strings of total length $N$, and compactly describes a collection of strings obtained by first choosing exactly
one string in every set, and then concatenating them together. This is motivated by the need of storing a collection of highly similar
DNA sequences.

The basic algorithmic question on elastic-degenerate strings is pattern matching: given such an elastic-degenerate string and a standard pattern of length
$m$, check if the pattern occurs in one of the strings in the described collection. Bernardini et al.~[SICOMP 2022] showed
how to leverage fast matrix multiplication to obtain an $\Ohtilde(nm^{\omega-1})+\Oh(N)$-time complexity for this problem, where $w$ is the matrix multiplication exponent.
However, from the point of view of possible applications, it is more desirable to work with approximate pattern matching,
where we seek approximate occurrences of the pattern. This generalisation has been considered in a few papers already,
but the best result so far for occurrences with $k$ mismatches, where $k$ is a constant, is the $\Ohtilde(nm^{2}+N)$-time algorithm of Pissis et al.~[CPM 2025]. This brings the question whether increasing the dependency on $m$ from
$m^{\omega-1}$ to quadratic is necessary when moving from $k=0$ to larger (but still constant) $k$.

We design an $\Ohtilde(nm^{1.5}+N)$-time algorithm for pattern matching with $k$ mismatches in an elastic-degenerate
string, for any constant $k$. To obtain this time bound, we leverage the structural characterization of occurrences with
$k$ mismatches of Charalampopoulos et al.~[FOCS 2020] together with fast Fourier transform. We need to work with
multiple patterns at the same time, instead of a single pattern, which requires refining the original characterization. This
might be of independent interest.
\end{abstract}

\clearpage
\setcounter{page}{1}

\section{Introduction}

An \emph{elastic-degenerate} string (ED-string, in short) $\T$ is a sequence $\T=\T[0] \T[1] \dd \T[n-1]$ of $n$ finite sets, where $\T[i]$ is a subset of $\Sigma^*$
and $\Sigma$ is an ordered finite alphabet. 
The \emph{length} of $\T$ is defined as the length $n=|\T|$ of the associated sequence.
The \emph{size} of $\T$ is defined as $N=N_{\stringeps}+\sum^{n}_{i=1}\sum_{S\in \T[i]} |S|$, where $N_{\stringeps}$ 
is the total number of empty strings in $T$. 
The \emph{cardinality} of $\T$ is defined as $c=\sum_{i=1}^n |\T[i]|$.
Intuitively, every ED-string represents a collection of strings, each of generally different length.
We formalize this intuition as follows. The \emph{language} $\L(\T)$ generated by $\T$ is defined 
as $\L(\T)=\{S_1 S_2 \dd S_n\,:\,S_i \in \T[i]\text{ for all }i\in[0, n-1]\}$.

The main motivation behind introducing ED-strings~\cite{DBLP:journals/iandc/IliopoulosKP21}
was to encode a collection of highly similar DNA sequences in a compact form.
Consider a multiple sequence alignment (MSA) of such a collection (see \cref{fig:EDS}).
Maximal conserved substrings (conserved columns of the MSA) form singleton sets of the ED-string
and the non-conserved ones form sets that list the corresponding variants.
Moreover, the language of the ED-string consists of (at least) the underlying sequences of the MSA.
Under the assumption that these underlying sequences are highly similar, the size of
the ED-string is substantially smaller than the total size of the collection.
ED-strings have been used in several applications: indexing a pangenome~\cite{BINF}, 
on-line string matching in a pangenome~\cite{DBLP:journals/bioinformatics/CislakGH18}, 
and pairwise comparison of pangenomes~\cite{FBINF}.

\begin{figure}[htpb]
    \centering
    \includegraphics[width=10cm]{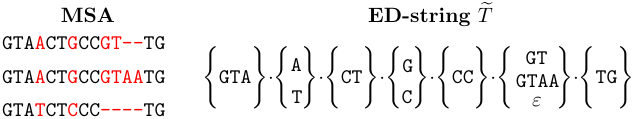}
    \caption{MSA of three strings and the corresponding (non-unique) ED-string $\T$.}
    \label{fig:EDS}
\end{figure}

ED-strings are also interesting as a simplified model
for string matching on node-labeled graphs~\cite{DBLP:conf/wabi/Ascone0CEGGP24}.
The ED-string $\T$ can be viewed as a graph of $n$ layers~\cite{DBLP:conf/wabi/MakinenCENT20}, 
where the $c$ nodes are strings from $\Sigma^*$, such that from layer $i-1$ to layer $i$ all possible edges are present 
and the nodes in layer $i$ are adjacent only to the nodes in layers $i-1$ and $i+1$.
As a simplified model, ED-strings offer important advantages, such as \emph{on-line} (left to right) 
string matching algorithms whose running times have a \emph{linear dependency on} 
$N$~\cite{DBLP:conf/cpm/GrossiILPPRRVV17,DBLP:conf/cpm/AoyamaNIIBT18,DBLP:journals/siamcomp/BernardiniGPPR22}.
This linear dependency on $N$ (without any multiplicative polylogarithmic factors) is highly desirable 
in applications because, nowadays, it is typical to encode a very large number of genomes 
(e.g., millions of SARS-CoV-2 genomes\footnote{\url{https://gisaid.org}}) in a single representation resulting in huge $N$ values.

In this work, we focus on the string matching (pattern matching) task.
In the \emph{elastic-degenerate string matching} (EDSM) problem, we are given a string $P$ of length $m$ (known as the \emph{pattern})
and an ED-string $\T$ (known as the \emph{text}), and we are asked to find the occurrences of $P$ in $\L(\T)$.
Grossi et al.~showed that EDSM can be solved in $\Oh(nm^2+N)$ time using combinatorial pattern matching tools~\cite{DBLP:conf/cpm/GrossiILPPRRVV17}. 
Aoyama et al.~improved this to $\Ohtilde(nm^{1.5})+\Oh(N)$ time by employing fast Fourier transform~\cite{DBLP:conf/cpm/AoyamaNIIBT18}. 
Finally, Bernardini et al.~improved it to $\Ohtilde(nm^{w-1})+\Oh(N)$ time~\cite{DBLP:journals/siamcomp/BernardiniGPPR22} 
by employing fast matrix multiplication, where $w<2.37134$~\cite{DBLP:conf/soda/AlmanDWXXZ25} is the matrix multiplication exponent. 
The authors of~\cite{DBLP:conf/soda/AlmanDWXXZ25} also showed a conditional lower bound for combinatorial algorithms\footnote{The 
term``combinatorial'' is not well-defined; such lower bounds often indicate that other techniques, including fast matrix multiplication, 
may be employed to obtain improved bounds for a specific problem.} for EDSM
stating that EDSM cannot be solved in $\Oh(nm^{1.5-\eps}+N)$ time, for any constant $\eps>0$.

In the approximate counterpart of EDSM, we are also given an integer $k>0$, and we are asked to find $k$-approximate occurrences of $P$ in $\L(\T)$; namely, the occurrences that are in Hamming or edit distance at most $k$ from $P$.
For Hamming distance, we call the problem \EDSMmis (see \cref{fig:EDSM}); and for edit distance, \EDSMerr.
The approximate EDSM problem was introduced by Bernardini et al.~who showed a simple $\Oh(k m c + kN)$-time algorithm 
for \EDSMmis and an $\Oh(k^2 m c + kN)$-time algorithm for \EDSMerr using combinatorial pattern
matching tools~\cite{DBLP:journals/tcs/BernardiniPPR20}. Pissis et al.~then improved the dependency on $k$ in both these results, 
obtaining an $\Ohtilde(k^{2/3}mc+\sqrt{k}N)$-time algorithm for \EDSMmis and an $\Ohtilde(kmc+kN)$-time 
algorithm for \EDSMerr~\cite{PRZ25}.

\begin{figure}[htpb]
    \centering
    \includegraphics[width=13cm]{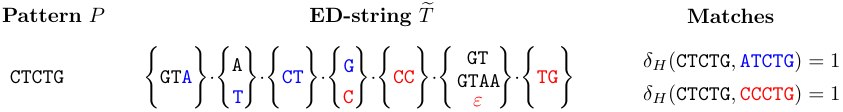}
    \caption{An example of \EDSMmis for $k=1$.}
    \label{fig:EDSM}
\end{figure}

Unfortunately, the cardinality $c$ of $\T$ in the above complexities is bounded only by $N$, so even for $k=1$, 
the existing algorithms run in $\Omega(mN)$ time in the worst case.
In response, Bernardini et al.~\cite{bernardini2024elastic} showed many algorithms for approximate EDSM for $k=1$ 
working in $\Ohtilde(nm^2+N)$ time or in $\Oh(nm^3+N)$ time for both the Hamming and the edit distance metrics. 
Pissis et al.~then improved these algorithms for $k=1$ (for both metrics) to $\Oh(nm^2+N)$ time~\cite{PRZ25}. 
The authors of~\cite{PRZ25} also showed algorithms for $k>1$ (for both metrics) 
that run in $\Ohtilde(nm^2 + N)$ time, for any \emph{constant} $k>1$. 
\renewcommand{\arraystretch}{1}
\begin{table}[h]
    \centering
    \begin{tabular}{c|c|c}
       \textbf{Time complexity} & \textbf{Remarks} & \textbf{Reference} 			\\\hline
       $\Oh(cm + N)$ & $k=\Oh(1)$ &  \cite{DBLP:journals/tcs/BernardiniPPR20} 	\\\hline
       $\Oh(nm^3+N)$ & $k=1$ & \cite{bernardini2024elastic} 						\\\hline
       $\Oh(nm^2+N\log m)$ & $k=1$ & \cite{bernardini2024elastic} 					\\\hline
       $\Oh(nm^2+N)$ & $k=1$ & \cite{PRZ25} 								\\\hline
       $\Ohtilde(nm^{1.5})+\Oh(N)$ & $k=1$ & \textbf{\cref{thm:main}} 				\\\hline
       $\Ohtilde(nm^2+N)$ & $k=\Oh(1)$ & \cite{PRZ25} 								\\\hline
       $\Ohtilde(nm^{1.5}+N)$ & $k=\Oh(1)$ & \textbf{\cref{thm:main2}} 				\\\hline
    \end{tabular}
    \caption{The upper-bound landscape of \EDSMmis for any constant $k\ge 1$.}
    \label{tab:1EDSM}
\end{table}

In this work, we consider the \EDSMmis problem with constant $k\ge 1$, and observe that all the existing algorithms
have at best a \emph{quadratic dependency on} $m$, the length of the pattern, for this problem (see \cref{tab:1EDSM}).
This is in stark contrast to the case of $k=0$, and brings the question of whether
non-combinatorial methods could be employed to solve \EDSMmis for any constant $k\ge 1$, 
similar to EDSM ($k=0$)~\cite{DBLP:conf/cpm/AoyamaNIIBT18,DBLP:journals/siamcomp/BernardiniGPPR22}. 

\begin{center}
	\emph{Can we break through the $m^2$ barrier for \EDSMmis for $k=\Oh(1)$?}
\end{center}


\begin{restatable}{theorem}{mainone}\label{thm:main}
Given a pattern $P$ of length $m$ and an ED-string $\T$ of length $n$ and size
$N$, \EDSMmis, for $k = 1$, can be solved in
$\Oh(nm^{1.5}\polylog m + N)$ time.
\end{restatable}

\begin{restatable}{theorem}{maintwo}\label{thm:main2}
Given a pattern $P$ of length $m$ and an ED-string $\T$ of length $n$ 
and size $N$, \EDSMmis, for any constant $k\ge 1$, can be solved in $\Oh((nm^{1.5}+N)\polylog m)$ time.
\end{restatable}

\paragraph{Previous approaches.} The key ingredient of~\cite{bernardini2024elastic} and~\cite{PRZ25} to achieve the $m^2$ factor
for $k=1$ is a new counterpart of the $k$-errata tree~\cite{coleetal}, where the copied nodes of the input trie are explicitly inserted into the tree.
This counterpart is an actual trie, and hence it allows to apply standard tree traversal algorithms. 
Since for $k=1$, the constructed trie for $\tilde{T}[i]$ and the suffixes of $P$, has $\cO(m\log N)$ nodes originating from $P$, 
bitwise $\cO(m/\log N)$-time operations per such node result in the desired complexity. 
The main tool in~\cite{PRZ25} for extending to a constant $k>1$ is also $k$-errata trees; however, the authors of~\cite{PRZ25} manage to apply $k$-errata trees as a black-box. We stress that those algorithms are combinatorial, that is, they do not use fast Fourier transform or fast matrix multiplication.

\paragraph{Our approach.}
As in the previous works on elastic-degenerate string matching, we work with the so-called active prefixes extension problem.
In this problem, we are given a text of length $m$, an input bitvector of length $m$, and a collection of patterns of total length $N$.
The goal is to produce an output bitvector, also of length $m$. Informally, whenever there is an occurrence of some pattern starting
at position $i$ and ending at position $j$, we can propagate a one from the input bitvector at position $i-1$ to a one in the output bitvector
at position $j$. For approximate pattern matching, we have $k+1$ input and output bitvectors, corresponding to matching prefixes
with different (i.e., the corresponding) number of mismatches.

The previous solutions can be seen as propagating the information from every $i-1$ to every $j$ explicitly. This, of course, cannot
achieve time complexity better than quadratic in $m$. Instead, we leverage the following high-level idea: if a given pattern occurs
very few times in the text, then we can afford to iterate through all of its occurrences and propagate the corresponding information.
Otherwise, its occurrences are probably somewhat structured. More concretely, exact occurrences of a pattern of length $\ell$
in a text of length $1.5\ell$ form a single arithmetical progression. This has been extended by Bringmann et al.~\cite{BringmannWK19}
to occurrences with $k$ mismatches, and further refined (and made effectively computable) by Charalampopoulos et al.~\cite{Charalampopoulos2020}:
either there are few occurrences, or they can be represented by a few arithmetical progressions (where few means polynomial in $k$).
Further, a representation of all the occurrences can be computed effectively.

To implement this high-level idea, we first apply some relatively standard preliminary steps that allow us to handle short patterns
effectively with $k$-errata trees. Further, we show how to reduce a given instance to multiple instances in which the pattern and the text
are roughly of the same length. Then, we handle patterns with only a few occurrences naively.
For the remaining patterns, we obtain a compact representation (as a few arithmetical progressions) of their occurrences.
We cannot afford to process each progression separately, but we observe that, because we have restricted the length of the text,
their differences are in fact equal to the same $q$ for every remaining pattern. Now, if $q$ is somewhat large
(with the exact threshold to be chosen at the very end to balance the complexities), we can afford to process every
occurrence of every pattern naively.
Otherwise, we would like to work with every remainder modulo $q$ separately, leveraging fast Fourier transform to
process all progressions starting at the positions with that remainder together.
As a very simplified example, if the progressions were the same for all the patterns, we only need to compute the sumset of the
set of starting positions with a one in the input bitvector (restricted to positions with a specific reminder modulo $q$)
with the set of lengths of the patterns. This can be indeed done with a single the fast Fourier transform.
However, the structural characterization of Charalampopoulos only says that, for every pattern, we have $\Oh(k^{2})$
arithmetical progressions with the same period $q$. However, the progressions are possibly quite different for
different patterns. Our new insight is that, in fact, we can group the progressions into only a few classes, more specifically
$\Oh(k^{2})$ irrespectively on the number of patterns, and then process
each class together. This requires looking more carefully at the structural characterization of Charalampopoulos et al.~\cite{Charalampopoulos2020},
and might be of independent interest.

\paragraph{Structure of the paper.}
In \cref{sec:prel}, we provide some preliminaries.
In \cref{sec:APE}, we discuss how \EDSMmis can be solved via the \APEmis problem, the auxiliary problem used also in previous solutions~\cite{DBLP:conf/cpm/AoyamaNIIBT18,DBLP:journals/siamcomp/BernardiniGPPR22,bernardini2024elastic,PRZ25}.
In \cref{sec:algo}, we present our algorithms for different cases: very short in \cref{sec:very short}; short in \cref{sec:short}; and
finally the most interesting case in \cref{sec:long}. We conclude with balancing the thresholds in \cref{sec:combine}.

\paragraph{Computational Model.} We assume the standard Word RAM model with words consisting of
$\Omega(\log n)$ bits, where $n$ is the size of the input. Basic operations on such words, such as indirect addressing and
arithmetical operations, are assumed to take constant time.

\section{Preliminaries}\label{sec:prel}

\paragraph{Strings.}
Let $\Sigma$ be a finite ordered alphabet of size $|\Sigma| = \sigma$.
We will usually assume that $\Sigma=\{1,2,\ldots,\poly(n)\}$, where $n$ is the size of the input, which is called the \emph{polynomial alphabet}.
Elements of $\Sigma$ are called \emph{characters}.
A sequence of characters from $\Sigma$, $X[0]X[1] \dd X[n-1]$, is called a (classic) \emph{string}
$X$.  We call $n$ the \emph{length} of $X$, and denote it by $|X|$.
The empty string is denoted by $\stringeps$.
By $X[i \dd j] = X[i]X[i+1] \dd X[j]$, we denote a \emph{fragment} of $X$
(starting at position $i$ and ending at position $j$), which equals $\stringeps$ when $i > j$.
Fragments of the form $X[\dd j] := X[0 \dd j]$ and $X[i \dd] := X[i \dd n-1]$
are called \emph{prefixes} and \emph{suffixes} of $X$, respectively.
A fragment of $X$ (or its prefix/suffix) is called \emph{proper}, if it is not equal to $X$.
Strings that are fragments of $X$ (for some $i$ and $j$) are called \emph{substrings} of $X$.
We also write $X[i \dd j)$ to denote $X[i \dd j-1]$.
By $XY$, we denote the \emph{concatenation} of strings $X$ and $Y$, i.e.
the string $X[0] \dd X[|X|-1]Y[0] \dd Y[|Y|-1]$. String $X$ is a \emph{cyclic shift} of $Y$ when $X=X_{1}X_{2}$ and $Y=X_{2}X_{1}$, and then we call $X$ and $Y$ \emph{cyclically equivalent}.
We say that $X$ is an \emph{occurrence} in $Y$ (at position $t$),
if $Y = AXB$ for some strings $A$ and $B$ such that $|A| = t - 1$.
Finally, $X^{r}$ is the \emph{reversal} of $X$, i.e. the string $X[n-1]X[n-2] \dd X[0]$.

\paragraph{Elastic-Degenerate Strings.}

We study the following extensions of classic strings.

A \emph{symbol} (over $\Sigma$) is an unordered subset of
(classic) strings from $\Sigma^*$, different from $\{\stringeps\}$ and $\emptyset$.
Note that symbols may contain $\stringeps$ but not as the only element.
The \emph{size} of a symbol is the total length of all strings in the symbol
(with the additional assumption that the empty string is counted as if it had
length 1).
The \emph{Cartesian concatenation} of two symbols $X$ and $Y$ is defined as
$X \otimes Y := \{ xy~|~x \in X, y \in Y \}$.

An \emph{elastic-degenerate string} (or \emph{ED-string})
$\X = \X[0]\X[1] \dd \X[n-1]$
(over alphabet $\Sigma$) is a sequence of symbols (over $\Sigma$).
We use $|\X|$ to denote the length of $\X$, i.e. the length of the associated sequence (the number of symbols).
The \emph{size} of $\X$ is the sum of the sizes of the symbols in $\X$.
As for classic strings, we denote by
$\X[i \dd j] = \X[i]\X[i+1] \dd \X[j]$
a \emph{fragment} of $\X$.
We similarly denote \emph{prefixes} and \emph{suffixes} of $\X$.
The \emph{language} of $\X$ is $\L(\X) = \X[0] \otimes \X[1] \otimes \dd \otimes \X[n-1]$.

Given a (classic) string $P$ and an ED-string $\T$, we say that
$P$ \emph{matches} the fragment $\T[i \dd j]$ (or that an \emph{occurrence}
of $P$ starts at position $i$ and ends at position $j$ of $\T$), if $i = j$
and $P$ is a fragment of at least one of the strings of $\T[i]$
(the whole pattern is fully contained in one of the symbols), or if $i < j$
and there is a sequence of strings $(P_i, P_{i+1}, \dd, P_j)$, such that
$P = P_i P_{i+1} \dd P_j$, $P_i$ is a suffix of one of the strings of
$\T[i]$, $P_k \in \T[k]$ for all $i < k < j$, and $P_j$ is a prefix of one
of the strings of $\T[j]$ (the pattern uses parts of at least two symbols).

\paragraph{Hamming distance.}
Given two (classic) strings $X, Y$ of the same length
over the alphabet $\Sigma$, their \emph{Hamming distance} $\dH(X, Y)$ is defined
as the number of \emph{mismatches} (positions $i$ such that $X[i] \neq Y[i]$).
We use $\Mis(X, Y)$ to denote the set of mismatches.

We say that a string $X$ is an \emph{approximate fragment}
(with at most $k$ mismatches) of $Y$ if there is a string $X'$ with
$\dH(X, X') \le k$, such that $X'$ is a substring of $Y$.
We similarly define \emph{approximate prefixes} and \emph{approximate suffixes}.
We write $\Occ(X, Y)$ to denote the set of all $k$-mismatch approximate occurrences of $X$ in $Y$,
i.e. all positions $i$, such that $\dH(X, Y[i \dd i + |X|)) \le k$.

Given a string $P$, an ED-string $\T$ and an integer $k \ge 1$, we say that
$P$ \emph{approximately matches} the fragment $\T[i \dd j]$
(with at most $k$ mismatches) of $\T$, or that an approximate occurrence of $P$
starts at position $i$ and ends at position $j$ of $\T$), if there is a string $P'$ such that
$\dH(P, P') \le k$ and $P'$ matches $\T[i \dd j]$.
We stress that, as in the case of exact occurrences, each approximate occurrence
of $P$ in $\T$ is of the following
forms: either $P$ has Hamming distance at most $k$ to a fragment
of a string in a symbol of $\T$; or
it uses parts of at least two symbols of $\T$. In the latter case, a prefix of $P$ is an approximate (possibly empty) suffix of a string in the first used symbol of $\T$,
a suffix of $P$ is an approximate (possibly empty) prefix of a string
in the last used symbol of $\T$,
and appropriate fragments of the pattern are approximate matches of a string in
all other used symbols of $\T$ (except the first and the last one).

\paragraph{Periodicity.}
For a string $X$, we write $X^{\infty}$ to denote the string $X$ concatenated infinitely many times with itself. We call a string $X$  \emph{primitive} when it cannot be represented as $Y^{h}$, for some string $Y$ and some integer $h\geq 2$.
We say that a string $X$ is a \emph{$d$-period with offset $\alpha$} of some other string $Y$ when
$ \dH(Y, X^\infty[\alpha\dd \alpha + |Y|)) \le d $
for some $d, \alpha \in \mathbb Z_{\ge 0}$;
and we call the elements of
$ \Mis(Y, X^\infty[\alpha\dd \alpha + |Y|)) $
the \emph{periodic mismatches}.
If $d = 0$, then $X$ is an \emph{exact period} of $Y$, or just a period.
Note that all cyclic shifts of $X$ are also (approximate or exact) periods, but with different offsets.

\paragraph{ED-string matching.}

Similarly as in \cite{bernardini2024elastic}, we define the following problem. Note that we assume integer $k$ to be a fixed constant
and not a part of the input.

\genproblem{Elastic Degenerate String Matching (EDSM) with $k$ Mismatches}{
  A string $P$ of length $m$ and
  an ED-string $\T$ of length $n$ and size $N \ge m$.
}{All positions in $\T$ where at least one approximate occurrence
(with at most $k$ mismatches) of $P$ ends.}

\noindent In the above problem, we call $P$ the \emph{pattern} and $\T$ the \emph{text}.

\paragraph{Active Prefixes Extension.}

As in previous works (cf.~\cite{bernardini2024elastic}), we solve \EDSMmis through the following auxiliary problem.

\genproblem{Active Prefixes Extension (APE) with $k$ Mismatches}{
A string $T$ of length $m$, $k+1$ bitvectors $U_{0}, U_{1}, \ldots, U_{k}$ of size $m$ each, and strings $P_{1},P_{2},\ldots$
of total length $N$.
}{$k+1$ bitvectors $V_{0}, V_{1}, \ldots, V_{k}$ of size $m$ each, where
$V_{d'}[j'] = 1$ if and only if there is a string $P_{i}$
and $j \in \{0, 1, \dd, m-1\}$ with $U_{d}[j] = 1$ such that $j' = j + |P_{i}|$ and $d' \ge d+\dH(P_i, T[j \dd j'])$.}

\noindent In the above problem, we call $T$ the \emph{text}, and $P_{1}, P_{2}, \ldots$ the \emph{patterns}.

\section{\EDSMmis via \APEmis}\label{sec:APE}

We begin by showing how to reduce \EDSMmis to multiple instances of \APEmis. This does not require any new ideas,
and proceeds similarly as in previous works (cf.~\cite{bernardini2024elastic}), so we only state it for completeness.

As we mentioned before, each approximate occurrence of pattern $P$ in ED-string
$\T$ is:
\begin{enumerate}
  \item either an approximate fragment of a string of a symbol; or
  \item crossing the boundary between two consecutive symbols.
\end{enumerate}
We explain how to detect the occurrences of each form separately.

\paragraph{Approximate fragments of symbols.}
To check if the pattern is an approximate fragment of a string of a symbol, we test each symbol of $\T$ separately (cf.~\cite{DBLP:journals/tcs/BernardiniPPR20}).
To this end, we apply the technique of Landau and Vishkin~\cite{LandauV86}, informally referred to as the \emph{kangaroo jumps}.
First, we preprocess the concatenation of all the symbols of $\T$ and the pattern with the following.

\begin{lemma}[suffix tree~\cite{Farach} with LCA queries~\cite{HarelT84}]
\label{lem:lce}
A string $T$ over a polynomial alphabet can be preprocessed in $\Oh(|T|)$ time to
allow computing the longest common prefix of any two suffixes $T[i\dd]$ and $T[j\dd]$ of $T$ in constant time.
\end{lemma}

\noindent Recall that pattern $P$ is of length $m$. For a symbol $X=\{ X_{1}, X_{2},\ldots,X_{t}\}$, we consider each string $X_{i}$ and, for every $j=0,1,\ldots,|X_{i}|-m$,
check if $\dH(X_{i}[j\dd j+m),P) \leq k$ in $\Oh(k)$ time by repeatedly computing the longest common prefix of the
remaining suffix of $X_{i}$ and $P$. This takes $\Oh(m+kN)$ time.

\paragraph{Crossing the boundary between two consecutive symbols.}
To check if the pattern approximately matches a fragment $\T[i\dd j]$, for some positions $i<j$, we reduce the problem to multiple instances of \APEmis.
We iterate through the symbols of $\T$ left-to-right and maintain
$k + 1$ bitvectors $B_0, B_{1}, \ldots, B_{k}$, each of size $m$,
such that $B_{d}[j] = 1$ when $P[\dd (j-1)]$ is a $d$-approximate suffix of
the current $\T[\dd i]$, for $d=0,1,\ldots,k$. Let $N_{i}$ be the size of $\T_{i}$, i.e. the total length of all strings in $\T_{i}$.

To proceed to the next iteration, and compute the bitvectors for
$\T[\dd i+1]$ from the bitvectors for $\T[\dd i]$, we need to consider two possibilities.
First, to consider the case when the $d$-approximate suffix is fully within $\T[i+1]$,
for every $d=0,1,\ldots,k$, we find all $d$-approximate prefixes of $P$ that are suffixes of $\T[i+1]$.
This is done in $\Oh(kN_{i+1})$ time by iterating over all strings in $\T[i+1]$, considering for each of them
every sufficiently short prefix of $P$, and computing the number of mismatches if it does not exceed $k$ using
kangaroo jumps in $\Oh(k)$ time.  
Second, to consider the case when the $d$-approximate suffix crosses the boundary between $\T[i]$ and $\T[i+1]$,
we create and solve an instance of \APEmis with the bitvectors representing the results for $\T[\dd i]$
and the strings in $\T[i+1]$.
We take as the new bitvectors the bitwise-OR of the bitvectors corresponding to both cases.

Before proceeding to the next iteration, we need to detect an occurrence that crosses the boundary between $\T[i]$ and $\T[i+1]$.
To this end, we consider each string $T\in \T[i+1]$. Then, for every $d=0,1\ldots,k$ and $j=0,1,\ldots,m-1$
such that $B_{d}[j]=1$ and $m-j \leq |T|$, we check if $P[j\dd]$ is a $(k-d)$-approximate prefix of $\T[i+1]$
using kangaroo jumps in $\Oh(k)$ time, and if so report position $i+1$ as a $k$-approximate occurrence.
Because we only need to consider $|T|+1$ possibilities for $j$, this takes $\Oh(kN_{i+1})$ time.

We summarise the complexity of the reduction in the following lemma.

\begin{lemma}
\label{lem:reduction}
Assume that \APEmis can be solved in $f_{k}(m,N)$ time. Then \EDSMmis can be solved in $\Oh(m+k\sum_{i}N_{i}+\sum_{i}f_{k}(m,N_{i}))$ time.
\end{lemma}

\section{Faster \APEmis}
\label{sec:algo}

We now move to designing efficient algorithms for \APEmis, separately for $k=1$ and then any constant $k$.
Combined with the reduction described in Lemma~\ref{lem:reduction}, this will result in \Cref{thm:main} and \Cref{thm:main2}.
Recall that the input to an instance of \APEmis consists of a string $T$ (called the text) of length $m$ and a collection of strings
$P_{1},P_{2},\ldots$ (called the patterns) of total length $N$.

For $k=1$, the strings are partitioned depending on their lengths and parameters $B'$ and $B$ depending on $m$:
\begin{enumerate}
    \item Very Short Case: the length of each string is $\leq B'$,
    \item Short Case: the length of each string is $>B'$ and $\leq B$,
    \item Long Case: the length of each string is $>B$.
\end{enumerate}
We separately solve the three obtained instances of \APEmis and return the bitwise-OR of the obtained bitvectors.

For an arbitrary constant $k$, we will have only two cases:
\begin{enumerate}
    \item Short Case: the length of each string is $\leq B$,
    \item Long Case: the length of each string is $>B$.
\end{enumerate}

\subsection{Very Short Case (for $k=1$)}
\label{sec:very short}

Recall that this case is only required for $k=1$.
Before stating the algorithm, we need a few standard tools.

\paragraph{Suffix tree.}
A \emph{trie} is a (rooted) tree, where every edge is labeled with a single character.
Each node of a trie represents the string obtained by concatenating the labels on its path from the root.
We consider only deterministic tries, meaning that the labels of all edges outgoing from the same node are pairwise distinct.
Then, a \emph{compact trie} is obtained from a trie by collapsing maximal downward paths on which every inner
node has exactly one child. The remaining nodes called \emph{explicit}, and the nodes that have been removed
while collapsing the paths are called \emph{implicit}. In a compact trie, every edge is labeled with a nonempty string, and the first characters
of all edges outgoing from the same node are pairwise distinct.

The \emph{suffix tree} of a string $T[0\dd n-1]$ is the compact trie of all the suffixes of $T\$$, where $\$$ is a special character
not occurring anywhere in $T$~\cite{Weiner1973}. Thus, there are $n+1$ leaves in the suffix tree of $T$, and thus contains $\Oh(n)$
nodes and edge. The label of each edge of the suffix tree is equal to some fragment $T[i\dd j]$, and we represent it by storing $i$ and $j$,
thus the whole suffix tree needs only $\Oh(n)$ space. For constructing the suffix tree we apply the following result.

\begin{lemma}[\cite{Farach}]
\label{lem:suffix}
The suffix tree of a string $T[0\dd n-1]$ over polynomial alphabet can be constructed in $\Oh(n)$ time.
\end{lemma}

The suffix tree of $T$, denoted $ST_{T}$, allows us to easily check if any string $X$ is a substring of $S$ by starting at the root and simulating
navigating down in the trie storing the suffixes of $T$. In every step, we need to choose the outgoing edge labeled by the
next character $X[i]$. If the current node is implicit, this is trivial to implement in constant time. Otherwise, we might
have multiple outgoing edges, and we need to store the first characters of their edges in an appropriate structure.
To this end, we use deterministic dictionaries.
 
\begin{lemma}\cite[Theorem 3]{Ruzic2008}
\label{lem:dictionary}
Given a set $S$ of $n$ integer keys, we can build in $\Oh(n(\log\log n)^{2})$ time a structure of size $\Oh(n)$, such that
given any integer $x$ we can check in $\Oh(1)$ time if $x\in S$, and if so return its associated information.
\end{lemma}

We apply Lemma~\ref{lem:dictionary} at each explicit node of the suffix tree. This takes $\Oh(n(\log\log n)^{2})$ total time,
and then allows us to implement navigating down in $\Oh(|X|)$ total time. This also gives us, for each prefix $X[\dd i]$,
its unique identifier: if we are at an explicit node then it is simply its preorder number, and otherwise it is a pair consisting of the
preorder number of the nearest explicit ancestor and the length of the current prefix. If in any step we cannot proceed further,
we set the identifier to null denoting that the prefix does not occur in $S$.
Such identifiers have the property that the identifier of $X[\dd i]$ is null if and only if $X[\dd i]$ does not occur in $S$,
and the identifiers of $X[\dd i]$ and $Y[\dd j]$ that both occur in $S$ are equal if and only if if the strings themselves are equal.
Further, we can think that each identifier is an integer from $\{1,2,\ldots,n^{2}\}$.

\begin{restatable}[]{theorem}{veryshort}
\label{thm:veryshort}
An instance of \APEmis where $k=1$ and the length of each pattern is at most $B'$ can be solved in $\Oh(m(B')^{2}+N)$ time.
\end{restatable}

\begin{proof}
We assume that the length of each pattern $P_{i}$ is at most $B'$.
Recall we are given bitvectors $U_{0}, U_{1}$ and the goal is to compute the bitvectors $V_{0}, V_{1}$.
This will be done by explicitly listing all fragments $T[j \dd j')$ such that $\dH(T[j \dd j'), P_{i}) = 0$, for some $i$,
and $\dH(T[j \dd j'), P_{i}) = 1$, for some $i$. For each such a fragment, we propagate the appropriate information from
the input bitvectors to the output bitvectors.

We begin with constructing the suffix trees of $T$ and $T^{r}$, including the dictionary structures
at each explicit node. Then, we distinguish two cases as follows.

\paragraph{Exact occurrences.}
For each string $P_{i}$, we find its identifier in $ST_{T}$ in $\Oh(|P_{i}|)$ time. If the identifier is non-null then
we include it in the set $S$.

We iterate over every position $j=0,1,\ldots,m-1$ and length $\ell=1,2,\ldots,\min\{B',m-j\}$.
While we iterate over the lengths, we simultaneously navigate in $ST_{T}$ to maintain the identifier of
$T[j\dd (j+\ell))$. To check if $T[j\dd (j+\ell)) = P_{i}$ for some $i$, we thus need to check if a specific identifier
belongs to $S$. Recall that the identifiers are integers from $\{1,2,\ldots,m^{2}\}$.
To avoid the necessity of paying extra logarithmic factors or using randomization, we answer all such queries together.
In more detail, we gather all the queries. Then, we sort the elements of $S$ and the queries together
with radix sort. Then, we scan the obtained sorted list and obtain the answer to each query in linear total time.
Finally, if $T[j\dd (j+\ell)) = P_{i}$ for some $i$ and $U_{d}[j]=1$ then we set $V_{d}[j]=1$, for $d=0,1$.

\paragraph{Occurrences with one mismatch.}
For each string $P_{i}$, we iterate over every position $j=0,1,\ldots,|P_{i}|-1$, assuming that the mismatch is there.
We would like to have access to the identifier of $P_{i}[\dd (j-1)]$ and $P_{i}[(j+1) \dd]$. This can be guaranteed
by first navigating in $ST_{T} $to compute the identifier of every prefix $P_{i}[\dd j]$ in $\Oh(|P_{i}|)$ time,
and similarly navigating in $ST_{T^{r}}$ to compute the identifier of the reversal of every suffix $(P_{i}[j\dd])^{r}$.
After such a preliminary step, for every position $j=0,1,\ldots,|P_{i}|-1$, if both identifiers are non-null then
we form a pair consisting of the identifier of $P_{i}[\dd (j-1)]$ and the identifier of $(P_{i}[(j+1) \dd])^{r}$. Let $S$
denote the obtained set of pairs.

We iterate over every position $j=0,1,\ldots,m-1$, $j'=j,j+1,\ldots,m-1$ and $j''=j',j'+1,\ldots,m-1$,
where $T[j\dd j'']$ is the considered fragment and $j'$ is the position of the mismatch.
We would like to have access to the identifier of $(T[j\dd j'))^{r}$ in $ST_{T_{r}}$ and the identifier of $T(j' \dd j'']$ in $ST_{T}$.
This can be assumed without increasing the time complexity by first iterating over $j'$ (in any order),
then over $j''$ in the increasing order, and finally over $j'$ in decreasing order, all while simultaneously
navigating in $ST_{T}$ and $ST_{T^{r}}$, respectively.
With the identifiers at hand, we need to check if the pair consisting of the identifier of $T(j' \dd j'']$ and
the identifier of $(T[j\dd j'))^{r}$ belongs to $S$. 
Similarly as for exact occurrences, this is done by answering the queries together with radix sort.
Then, if $\dH(T[j\dd j''],P_{i}) \leq 1$ for some $i$ and $U_{0}[j]=1$ then we set $V_{1}[j]=1$.

\paragraph{Summary.}
The algorithm described above consists of the following steps.
First, we need to construct the suffix trees of $T$ and $T^{R}$ in $\Oh(m)$ time. Constructing the deterministic dictionaries
storing the outgoing edges takes $\Oh(m(\log\log m)^{2})$ time.
Second, listing and processing the exact occurrences takes $\Oh(m B' + N)$ time.
Third, listing and processing occurrences with one mismatch takes $\Oh(m(B')^{2} + N)$ time.
\end{proof}

\subsection{Short Case}
\label{sec:short}

Recall that we are given the patterns $P_{1},P_{2},\ldots,P_{d}$ of total length $N$ and the text $T[0\dd (m-1)]$.
Further, the length of each pattern $P_{i}$ is at least $B'$ but at most $B$.
We start with observing that, after $\Oh(m+N)$-time preprocessing, we can assume that $\log d = \Oh(k\log m)$,
because we only need to keep patterns $P_{i}$ such that $\dH(T[j\dd j'],P_{i})\leq k$, for some fragment $T[j\dd j']$. This relates the number $d$ of the patterns to $k$ and $m$. Then, we state another known tool, and
finally provide the algorithm.

\paragraph{Reducing the number of patterns.}
Recall that we only need to keep patterns $P_{i}$ such that $\dH(T[j\dd j'],P_{i})\leq k$
for some fragment $T[j\dd j']$. To reduce the number of patterns, we first run the following procedure.

\begin{lemma}
\label{lem:partition}
Assuming access to the suffix tree of $T[0\dd (m-1)]$, given a pattern $P[0\dd (n-1)]$ we can check if
$\dH(T[j\dd j'],P)\leq k$ for some fragment $T[j\dd j']$, and if so represent $P$ as a concatenation of $2k+1$
fragments of $T$ and special characters not occurring in $T$.
\end{lemma}

\begin{proof}
We repeat the following step $2k+1$ times, starting with the whole $P[0\dd (n-1)]$.
For the current suffix
$P[j\dd]$, we navigate the suffix tree of $T$ starting from the root to find its longest prefix $P[j\dd j']$ equal to some substring
of $T$. Then, we cut off $P[j\dd j']$ and repeat.
As a special case, if $P[j]$ does not occur in $T$, we cut off the first character of $P[j\dd]$. 
We claim that if after $k+1$ steps the remaining suffix is non-empty then $P$ cannot be within Hamming distance at most $k$
from some fragment of $T$.
Assume otherwise. Each character not occurring in $T$ incurs one mismatch.
Each maximal range of fragments of $T$ created by our procedure, denoted by $F_{1} F_{2}\ldots F_{t}$,
has the property that each $F_{i}$ is a fragment of $T$, but $F_{j} F_{j+1}[1]$ is not (for $j<t$).
Thus, each $F_{j} F_{j+1}[1]$ incurs at least one mismatch, and those mismatches are disjoint for $j=1, 3, 5, \ldots$.
Summing up, if the number of mismatches is at most $k$ then the number of steps is at most $2k+1$.
\end{proof}

We build the suffix tree of $T$ with Lemma~\ref{lem:suffix}, and then apply Lemma~\ref{lem:partition} on every pattern
$P_{i}$ in $\Oh(m+N)$ total time. Then, we only keep patterns $P_{i}$ such that $\dH(T[j\dd j'],P)\leq k$ for some fragment $T[j\dd j']$.
Further, each such pattern is represented as a concatenation of $\Oh(k)$ fragments of $T$ and special characters not occurring
in $T$, and thus can be described by specifying $\Oh(k)$ integers from $\{0,1,\ldots,m\}$, with every fragment represented
by two such integers, and every character not occurring in $T$ represented by $0$.
Among all patterns with the same representation, it is enough to keep only one, which ensures
that $d=\Oh(m^{\Oh(k)})$, so $\log d=\Oh(k\log m)$.

\paragraph{The $k$-errata trie.}
Cole, Gottlieb, and Lewenstein~\cite{coleetal} considered the problem of preprocessing a dictionary of $d$ patterns $P_{1},P_{2},\ldots,P_{d}$
of total length $N$ for finding, given a query string $Q$ of length $m$, whether $Q$ is at Hamming distance at most $k$ from some $P_{i}$.
We provide a brief overview of their approach, following the exposition in~\cite{gawrychowskihal} that provides some details
not present in the original description.

For $k=0$, this can be of course easily solved with a structure of size $\Oh(N)$ and query time $\Oh(m)$ by arranging the patterns
in a trie. For larger values of $k$, the $k$-errata trie is defined recursively. In every step of the recursion, the input is a collection
of $x \leq d$ strings, each of them being a suffix of some pattern $P_{i}$ decorated with its mismatch budget initially set to $k$.
We arrange the strings in a compact trie, and then recurse guided by the heavy-path decomposition of the trie.
The depth of the recursion is $k$, and on each level the overall number of strings increases by a factor of $\log d$, starting from $d$.
Answering a query requires the following primitive: given a node of one of the compact tries and the remaining suffix of the query string
$Q[i\dd]$, we need to navigate down starting from the given node while reading off the subsequent characters of $Q[i\dd]$. This needs
to be done while avoiding explicitly scanning $Q[i\dd]$, as such a primitive is invoked multiple times. For a compact trie
storing $x$ suffixes of the patterns, such a primitive can be implemented by a structure of size $\Oh(x\log x)$ and query time
$\Oh(\log\log N)$, assuming that we know the position of every suffix of the query string in the suffix tree of
$P_{1}\$_{1} P_{2}\$_{2} \ldots P_{d}\$_{d}$ (also known as the generalized suffix tree of $P_{1}, P_{2}, \ldots, P_{d}$).

In our application, the query string will be always a fragment of the text $T[i\dd j]$. Thus, we can guarantee that the position
of every suffix of the query string in the generalized suffix tree of $P_{1}, P_{2}, \ldots, P_{d}$ is known by building the generalized
suffix tree of $P_{1}, P_{2}, \ldots, P_{d}, T$. This gives us the position of every suffix $T[i\dd]$ in the generalized suffix tree
$P_{1}, P_{2}, \ldots, P_{d}$, from which we can infer the position of $T[i\dd j]$.
We summarise the properties of such an implementation below.

\begin{lemma}[\cite{coleetal}]
\label{lem:errata}
For any constant $k$, a dictionary of $d$ patterns $P_{1},P_{2},\ldots,P_{d}$ of total length $N$ and a text $T[0\dd (m-1)]$ can be
preprocessed in $\Oh(m+N+d\log^{k+1}d)$ time to obtain a structure of size $\Oh(m+N+d\log^{k}d)$, such that
for any fragment $T[i\dd j]$ we can check in $\Oh(\log^{k}d\log\log N)$ time whether $\dH(T[i\dd j],P_{i}) \leq k$, for some $i$.
\end{lemma}

\begin{theorem}
\label{thm:short}
For any constant $k$, an instance of \APEmis where the length of each pattern is at least $B'$ and at most $B$ can be solved in
$\Oh(mB\log^{k}m\log\log m+N+N/B'\cdot\log^{k+1}m)$ time.
\end{theorem}

\begin{proof}
We start with applying Lemma~\ref{lem:errata} on the patterns $P_{1},P_{2},\ldots$ and the text $T[0\dd (m-1)]$. Then, iterate over
every position $j=0,1,\ldots,m-1$, length $\ell=1,2,\ldots,\min\{m-j,B\}$, and $d=0,1,\ldots,k$ such that $U_{d}[j]=1$ .
Next, for every $d'=0,1,\ldots,k-d$ we check if $\dH(T[i\dd (i+\ell)),P_{i})\leq d'$ for some $i$. If so, we set $V_{d+d'}[j+\ell]=1$.

We analyse the overall time complexity. First, we need to construct the $k$-errata trie for
$P_{1},P_{2},\ldots,P_{d}$ and $T[0\dd (m-1)]$. This takes $\Oh(m+N+d\log^{k+1}d)$ time.
Then, we consider $\Oh(mB)$ possibilities for iterating over the position and the length, and for each of them spend $\Oh(\log^{k}d\log\log N)$
time.
As each $P_{i}$ is of length at least $B'$, from the preprocessing we have $\log d=\Oh(k\log m)$ and $N\leq dm$, the overall complexity is:
\[
\Oh(m+N+d\log^{k+1}d +mB\log^{k}d\log\log N) = \Oh(N+N/B'\log^{k+1}m + mB\log^{k}m\log\log m)
\]
as claimed.
\end{proof}

\subsection{Long Case}
\label{sec:long}

In the most technical case, we assume that the length of each pattern $P_{i}$ is at least $B$.
We start with providing an overview, and then move to filling in the technical details.

The very high-level idea is to explicitly or implicitly process all occurrences of every pattern $P_{i}$. If a given pattern $P_{i}$
occurs sufficiently few times in the text then we can afford to list and process each of its occurrences explicitly. Otherwise,
we invoke the structural characterization of~\cite{Charalampopoulos2020}, which, roughly speaking, says that if there are
many approximate occurrences of the same string sufficiently close to each other in the text, then the string and the relevant fragment of
the text have a certain regular structure. Thus, we can certainly hope to process all occurrences of such a pattern $P_{i}$ together
faster than by considering each of those occurrences one-by-one. However, this would not result in a speed-up, and in fact, we need
to consider multiple such patterns $P_{i}$ together. To this end, we need to further refine the characterization
of~\cite{Charalampopoulos2020}. Before we proceed with a description of our refinement, we start
with a summary of the necessary tools from~\cite{Charalampopoulos2020}. Then, we introduce some notation,
introduce some simplifying assumptions, and then describe our refinement.

\paragraph{Tools.}
The authors of~\cite{Charalampopoulos2020} phrase their algorithmic results using the framework of
\texttt{PILLAR} operations. In this framework, we operate on strings, each of them specified by a handle.
For two strings $S$ and $T$, the following operations are possible (among others):
\begin{enumerate}
\item \texttt{Extract}$(S,\ell,r)$: retrieve the string $S[\ell\dd r]$,
\item \texttt{LCP}$(S,T)$: compute the length of the longest common prefix of $S$ and $T$,
\item \texttt{IPM}$(S,T)$: assuming that $|T|\leq 2|S|$, return the starting positions of all exact occurrences of $S$ in $T$ (at most two starting positions or an arithmetical progression of starting positions).
\end{enumerate}

\begin{lemma}[{\cite[Theorem 7.2]{Charalampopoulos2020}}]
	\label{pillar simulation}
	After an $\Oh(N)$-time preprocessing of a collection of strings of total length $N$, each \texttt{PILLAR} operation can be performed in $\Oh(1)$ time.
\end{lemma}

\noindent We apply the above lemma on the text and all the patterns in $\Oh(m+N)$ time.

The first main result of~\cite{Charalampopoulos2020} is the following structural characterization.

\begin{lemma}[{\cite[Theorem 3.1]{Charalampopoulos2020}}]
	\label{two cases pillar theorem}
	For each pattern of length $|P| \le 1.5|T|$ at least one of the following holds:
	\begin{itemize}
		\item $|\Occ(P, T)| \le 864k$,
		\item There is a primitive string $Q$ of length $|Q| \le |P|/128k$ such that $\dH\left(P, Q^\infty\left[0 \dd |P|\right)\right) < 2k$.
	\end{itemize}
\end{lemma}

\noindent Then, they convert the structural characterization (\cref{two cases pillar theorem}) into an efficient algorithm.

\begin{lemma}[{\cite[Main Theorem 8]{Charalampopoulos2020}}]
	\label{calculate occ}
	For any pattern of length $|P| \le 1.5|T|$ we can compute (a representation of) $\Occ(P, T)$ in time $\Oh(k^2 \log \log k)$ plus $\Oh(k^2)$ \texttt{PILLAR} operations.
\end{lemma}

\noindent The representation is a set of $\Oh(k^{2})$ arithmetical progressions. Further, as the algorithm
follows the proof of Lemma~\ref{two cases pillar theorem}, in fact it either outputs a set of $\Oh(k)$ occurrences
or finds a primitive string $Q$ of length $|Q| \le |P|/128k$ such that $\dH\left(P, Q^\infty\left[0 \dd |P|\right)\right) < 2k$.

\paragraph{Notation.}
We rephrase the \APEmis problem as follows. For a pattern $P$ and a set of positions $A$ in the text $T$ we define:
\[ \Ext\left(P, T, A\right) := \left(\Occ \left(P, T\right) \cap A \right) + |P|,\]
which is also a set of positions in $T$. Then, \APEmis reduces to $\Oh(k^{2})$ instances of computing:
\[
	\bigcup_i \Ext[k'](P_i, T, A)
\] 
for a given set of patterns $P_{1},P_{2},\ldots$, a text $T$, a set of positions $A$, and parameter $k'\leq k$.
In every instance, the set $A$ is simply the characteristic vector of some bitvector $U_{d}$, and the obtained
set of positions is treated as another bitvector that contributes to some bitvector $V_{d'}$.
From now on, we focus on designing an algorithm that computes $\bigcup_i \Ext(P_i, T, A)$, and
identify the underlying sets with their characteristic vectors. For two sets of integers $X$ and $Y$, we define their sumset
$X \oplus Y := \{ x+y : x \in X, y \in Y\}$.  For a set of integers $x$ and a shift $s$, we define $X+s := \{ x + s : x\in X\}$. The following result is well-known.

\begin{lemma}[e.g.~\cite{Furer14a}]
\label{lem:fft}
Given $X,Y\subseteq [0,\ell)$, we can calculate $X \oplus Y$ in $\Oh(\ell \log \ell)$ time.
\end{lemma}

\paragraph{Simplifying assumptions.}
It is convenient to assume that each pattern has roughly the same length, similar to the length of the text.
More formally, our algorithm will assume that we have:
\begin{enumerate}
\item $|P_i| \in [\ell, 1.1 \ell)$, for every $i$,
\item $|T| \in [\ell,1.5 \ell]$,
\end{enumerate}
for some $\ell$. Any instance can be reduced $\Oh(\log m)$ instances in which $|P_i| \in [\ell, 1.1 \ell)$, for every $i$,
by considering $\ell = 1.1^0, 1.1^1, 1.1^2, \dots \le m$.
For each such $\ell \geq B$, we create a separate instance containing only patterns of length from $[\ell, 1.1 \ell)$.
As each pattern $P_i$ falls within exactly one such instance, designing an algorithm running in $\Oh(f(m)+N)$ time for every such
instance, implies an algorithm running in $\Oh(f(m)\log m+N)$ time for a general instance.
To additionally guarantee that $|T| \le 1.5 \ell$ (so that \cref{two cases pillar theorem} can be directly applied), we choose $|T| / 0.4 \ell$
fragments $T_j$ such that each potential occurrence of a pattern $P_i$ in $T$ falls within some fragments $T_j$.
Formally, if $T_j$ is the $1.5\ell$-length fragment (possibly shorter for the very last fragment) starting at position $0.4 j \ell$ then:
\[
	\Occ(P_i, T) = \bigcup_j \Occ(P_i, T_j) + 0.4j\ell,
\]
where we disregard fragments shorter than $\ell$ as they cannot contain an occurrence of any $P_{i}$.
From now on, always assume that we deal with a single text $T$, with $|T| \in [\ell, 1.5 \ell]$, 
and a set of patterns $P_i$ with lengths in $[\ell, 1.1 \ell)$ (we will sometimes omit the index of a pattern and simply write $P$).
The preprocessing from Lemma~\ref{pillar simulation} is performed only once, and then in each instance
we assume that any \texttt{PILLAR} operation can be performed in $\Oh(1)$ time.
The input bitvectors $U_{0},U_{1},\ldots,U_{k}$ in such an instance are fragments of the original input bitvectors,
and after computing the output bitvectors $V_{0},V_{1},\ldots,V_{k}$ we update appropriate fragments of the
original output bitvectors by computing bitwise-OR.
The final number of restricted instances is $\Oh(m/\ell\cdot \log m)$, and each original pattern appears in $\Oh(m/\ell)$ instances.

Consider a restricted instance containing $d$ patterns. Our goal will be to solve it in $\Oh((d + \ell \sqrt{d\log \ell})\poly(k))$ time.
Before we proceed to describing such an algorithm, we analyse what does this imply for an algorithm solving the original instance.

\begin{theorem}
\label{thm:long}
For any $k$, an instance of \APEmis where the length of each pattern is at least $B$ can be solved in
$\Oh(m+(Nm/B^{2}+m^{2}\log^{2}m/B)\poly(k))$ time.
\end{theorem}

\begin{proof}
We assume that a restricted instance with $d$ patterns can be solved in 
$\Oh((d + \ell \sqrt{d\log \ell})\poly(k))$ time, and describe an algorithm for solving a general instance
of \APEmis.

Let $d_{i}$ denote the number of patterns of length from $[\ell_{i},\ell_{i+1})$, where $\ell_{i}=1.1^{i}$, in the original instance.
Recall that we only consider $i$ such that $\ell_{i} \geq B$.
After the initial $\Oh(m+N)$-time preprocessing, ignoring factors polynomial in $k$, the total time to solve the restricted instances is:
\begin{align*}
\Oh(\sum_{i} m/\ell_{i}(d_{i}+\ell_{i}\sqrt{d_{i}\log\ell_{i}}) &= \Oh(\sum_{i}m/\ell_{i}^{2}\cdot d_{i}\ell_{i}+\sum_{i}m\sqrt{d_{i}\log m}) \\
& = \Oh(Nm/B^{2} + \sum_{i} m\sqrt{d_{i}\log m}),
\end{align*}
where we have used $\sum_{i} d_{i}\ell_{i} = N$ and $\ell_{i} \geq B$.
We split the sum by separately considering $i$ such that $m\sqrt{d_{i}\log m} \leq d_{i}\ell_{i}$, i.e. $\sqrt{d_{i}} \geq m\sqrt{\log m}/\ell_{i}$.
This gives us:
\begin{align*}
& \Oh(Nm/B^{2} + \sum_{i} d_{i}\ell_{i} + \sum_{i}m^{2}\log m/\ell_{i}) = \Oh(Nm/B^{2}+N+m^{2}\log^{2}m/B).
\end{align*}
Thus, as long as we indeed manage to solve a restricted instance in the promised complexity, we
obtain the theorem.
\end{proof}

In the remaining part of this section, we describe an algorithm for solving a restricted instance of \APEmis
containing $d$ patterns in $\Oh((d + \ell \sqrt{d\log \ell})\poly(k))$ time.

\paragraph{Additional assumptions.}
We start with applying Lemma~\ref{calculate occ} on every pattern $P_{i}$ to obtain a representation of its occurrences in the text
in $\Oh(d\poly(k))$ time.
As mentioned earlier, the algorithm either outputs a set of $\Oh(k)$ occurrences
or finds a primitive $2k$-period $Q_{i}$ of $P_{i}$ such that $|Q_{i}| \le |P|/128k < \ell / 100k$ (the second inequality holds
because $|P| < 1.1\ell$).
In the latter case, we also obtain a representation of the whole set of occurrences as $\Oh(k^{2})$
arithmetical progressions.

If there are $\Oh(k)$ occurrences of $P_{i}$ in the text then we process each of them naively in $\Oh(dk)$ time.
From now on we can thus assume otherwise for every pattern $P_{i}$.
Then, we consider the text, and ensure that it is fully covered by approximate occurrences of the patterns:
\begin{itemize}
	\item some pattern $P$ is a $k$-mismatch prefix of $T$, formally $\dH(P, T[0 \dd |P|)) \le k$; and
	\item some pattern $P'$ is a $k$-mismatch suffix of $T$, formally $\dH(P', T[|T| - |P'| \dd |T|)) \le k$.
\end{itemize}
This is guaranteed by removing some prefix and some suffix of the text; it can be implemented in $\Oh(d\poly(k))$ time
by extracting the first and the last occurrence from each arithmetical progression in the representation.
Then the following holds.

\begin{restatable}[]{lemma}{restateCommonPeriodTheorem}
	\label{we have common period theorem}
	All $Q_{i}$s are cyclically equivalent, and every $Q_{i}$ is a $6k$-period of the text.
\end{restatable}

\begin{proof}
	We first prove that the periods $Q_i$ obtained for all the patterns $P_i$ must be cyclically equivalent.
	Let $T_{mid} := T[\ell / 2 \dd \ell)$ be the middle section of $T$.
	Since all patterns are of length $|P| \ge \ell$ and the text is of length $|T| \le 1.5\ell$, all pattern occurrences must cover the middle section of $T$.
	Recall that we assume that every pattern $P_i$ has some $2k$-period $Q_i$.
	By triangle inequality, every $Q_i$ must be a $3k$-period of $T_{mid}$.
	We will first show that if the strings $Q_i$ are primitive and of length $|Q_i| < \ell / 100k$, then they are all cyclically equivalent.
	Select any two such periods of $T_{mid}$, denoted $Q_1$ and $Q_2$, and assume (only to avoid clutter) that both of
	their offsets are equal to $0$.
	
	First, assume that $Q_1$ and $Q_2$ are not of the same length.
	Observe that since the size of the combined set of mismatches
	$
	\Mis(T_{mid}, Q_1^\infty[0 \dd \ell / 2)) \cup 
	\Mis(T_{mid}, Q_2^\infty[0 \dd \ell / 2))
	$
	is at most $6k$, there must exist a substring $T_{sub}$ of $T_{mid}$ that does not contain any such mismatch of length at least
	\[
		|T_{sub}| \ge \left \lceil \frac{|T_{mid}| - 6k}{6k + 1} \right \rceil \ge \frac{\frac{1}{2} \ell + 1}{6k + 1} - 1 > \frac{\ell}{14k} - 1.
	\]
	The strings $Q_1$ and $Q_2$ are thus exact periods of $T_{sub}$. In addition we have
	\[
		|Q_1| + |Q_2| \le \ell / 100k + \ell / 100k < \ell / 14k < |T_{sub}| + 1
	\]
	which by the periodicity lemma of Fine and Wilf~\cite{fineWilf65} 
	induces a period of length $\gcd(|Q_1|, |Q_2|)$, and contradicts the assumption that $Q_1$ and $Q_2$ are primitive.
	
	In the other case, when $|Q_1| = |Q_2|$, assume that $Q_1 \neq Q_2$.
	We would then have
	\[
		\dH\left(Q_1^\infty\left[0 \dd \ell / 2\right), Q_2^\infty\left[0 \dd \ell / 2\right)\right)
		\ge \left\lfloor \frac{\ell / 2}{|Q_1|} \right \rfloor
		\ge \left\lfloor \frac{\ell / 2}{\ell / 100k} \right \rfloor
		= 50k.
	\]
	On the other hand by triangle inequality
	\[
		\dH\left(Q_1^\infty\left[0 \dd \ell / 2\right), Q_2^\infty\left[0 \dd \ell / 2\right)\right)
		\le \dH\left(Q_1^\infty\left[0 \dd \ell / 2\right), T_{mid}\right) + \dH\left(T_{mid}, Q_2^\infty\left[0 \dd \ell / 2\right)\right)
		\le 3k + 3k,
	\]
	which again gives us a contradiction and proves that $Q_1$ must be equivalent to $Q_2$.

	Now since we assume that some $P$ is a $k$-mismatch prefix of $T$ and some $P'$ is a $k$-mismatch suffix of $T$, both having $2k$-period $Q$, it can be proven with similar arguments that $Q$ is a $6k$-period of $T$.
\end{proof}

We choose $Q$ to be a cyclic shift of $Q_{1}$, so $Q$ is a $2k$-period of every pattern, and a $6k$-period with offset 0 of the text. This can be implemented in $\Oh(d\poly(k)+\ell)$ time as follows.
Because $|Q_{i}| < \ell / 100k$ and $\dH\left(P_{i}, Q_{i}^\infty\left[0 \dd |P_{i}|\right)\right) \leq 2k$,
in fact for some $j$ we have $P_{i}[j\cdot |Q_{i}| \dd (j+2)\cdot |Q_{i}|) = Q_{i}Q_{i}$. Further,
such a $j$ can be computed in $\Oh(k)$ time by just trying $j=0,1,2,\dd$, and verifying each
$j$ by computing the longest common prefix twice. Overall, this takes $\Oh(d\poly(k))$ time.
We start with setting $Q' := Q_{1}$. Then, we search for a cyclic shift $Q$ of $Q'$ such that
$\dH\left(T, Q^\infty\left[0 \dd |T|\right) \right) \leq 6k$. To this end, we check all possible $\Oh(\ell/k)$
cyclic shifts. To verify whether $Q$ is a good cyclic shift, we
extract the mismatches between $T$ and $Q^\infty\left[0 \dd |T|\right)$, terminating when there are
more than $6k$. The next mismatch can be found in constant time by first
computing the longest common prefix of the remaining suffix of $T$ with an appropriate cyclic shift
of $Q$, and if there is none by computing the longest common prefix of the remaining suffix of $T$
with the suffix shortened by $|Q|$ characters. Overall, this takes $\Oh(\ell)$ time.
After having found $Q$, we also compute, for every
pattern $P_{i}$, an integer $r$ such that $\dH\left(P_{i}, Q^\infty\left[r \dd r+|P_{i}|\right)\right) \leq 2k$,
which can be done with a single internal pattern matching query to find
an occurrence of $Q$ in $Q_{i}Q_{i}$.

\paragraph{The algorithm.}
To obtain an efficient algorithm, we will partition the set of all positions $[0 \dd |T|)$ into $\Oh(k)$ consecutive regions
$R_{0},R_{1},\ldots, R_{b}$ with the property that if we restrict the text to any region $R_i$, the corresponding fragment
is \emph{almost} periodic with respect to $Q$; more specifically, it may have a single periodic mismatch at the rightmost position.
Then, for each pair of regions $R_s$ and $R_t$, with $s\leq t$, we separately calculate the set of extensions induced by pattern occurrences
that start inside $R_s$ and end inside $R_t$:
\[
	\bigcup_i \Ext(P_i, T, A) = \bigcup_{s} \bigcup_{t} \bigcup_i \Ext(P_i, T, A \cap R_s) \cap R_t.
\]
Since $b=\Oh(k)$, this allows us to reduce the problem to $\Oh(k^2)$ separate instances of calculating:
\[
	\bigcup_i \Ext(P_i, T, A \cap R_s) \cap R_t.
\]

Consider a single pattern $P$, and recall that by the additional assumptions, we have a primitive string $Q$ of length $|Q| < \ell / 100k$
such that:
\begin{align*}
	\dH(P, \QP) \le 2k \qquad\text{ for } \QP &:= Q^\infty[r \dd r + |P|), \\
	\dH(T, \QT) \le 6k \qquad\text{ for } \QT &:= Q^\infty[0 \dd |T|).
\end{align*}
Further, let $C_{r}$ be the positions congruent to $r$ modulo $|Q|$ in the text.
We start with recalling from~\cite{Charalampopoulos2020} that the positions of all $k$-mismatch occurrences are congruent
modulo $|Q|$. We provide a proof for completeness.

\begin{restatable}[]{lemma}{restateOccSubseteqProgression}
	\label{occ subseteq progression}
	$\Occ(P, T) \subseteq \set{r + i|Q| \ : \ i \in \mathbb{Z}}$.
\end{restatable}
\begin{proof}
	For any $x \in \Occ(P, T)$ by triangle inequality we have
	\[
		\dH(\QP, \QT[x \dd x + |P|)) \le
		\dH(\QP, P) + 
		\dH(P, T[x \dd x + |P|)) +
		\dH(T[x \dd x + |P|), \QT[x \dd x + |P|))
		\]
		and since
		\begin{itemize}
			\item $\dH(\QP, P) \le 2k$,
			\item $x \in \Occ(P, T) \Rightarrow \dH(P, T[x \dd x + |P|)) \le k$,
			\item $\dH(T[x \dd x + |P|), \QT[x \dd x + |P|)) \le \dH(T, \QT) \le 6k$
		\end{itemize}
		we get
		\[ \dH(\QP, \QT[x \dd x + |P|)) \le 9k. \]
		Recall that $\QP = Q^\infty[r \dd r + |P|)$ and $\QT[x \dd x + |P|) = Q^\infty[x \dd x + |P|)$ both have a primitive period $Q$ (with offsets $r$ and $x$, respectively).
		If their offsets are not congruent modulo $|Q|$, we can bound the number of mismatches by
		\[ \dH(\QP, \QT[x \dd x + |P|)) \ge \left \lfloor |P| / |Q| \right \rfloor > 9k,\]
		which yields a contradiction (the second inequality follows from $|Q| \le |P| / 128k$).
		Therefore $x \in \set{r + i|Q| \ : \ i \in \mathbb{Z}}.$
\end{proof}

Following \Cref{occ subseteq progression}, choose $r$ such that $\Occ(P, T) \subseteq C_r$. 
In order to characterize $\Occ(P, T)$, let us now analyze the values $\dH(P, T[x \dd x + |P|))$ for $x \in C_r$.
From triangle inequality we have
\begin{equation}
	\dH(P, T[x \dd x + |P|)) \le \dH(P, \QP) + \dH(\QP, \QT[x \dd x + |P|)) + \dH(\QT[x \dd x + |P|), T[x \dd x + |P|)). \label{basic triangle inequality} 
\end{equation}
Observe that 
\begin{equation}
	\QP = \QT[x \dd x + |P|), \label{qp=qt}
\end{equation}
since both strings have a period $Q$ with offsets congruent modulo $|Q|$, which gives us
\begin{equation}
	\dH(P, T[x \dd x + |P|)) \le \dH(P, \QP) + \dH(\QT[x \dd x + |P|), T[x \dd x + |P|)). \label{triangle inequality}
\end{equation}
We will later show that the above inequality is in fact an equality for all $x \in C_r$ except for at most $\Oh(k^2)$ exceptions $E$.
Specifically, define
	\[ E := \set{\tau - \pi \ : \ \pi \in \Mis(P, \QP), \ \tau \in \Mis(T, \QT)}, \]
which intuitively represents the set of all starting positions $x$ in $T$ such that when comparing $P$ and $T[x \dd x + |P|)$, at least one pair of mismatches aligns with each other.
Note that $|E| \le \dH(P, \QP) \cdot \dH(T, \QT) \le 2k \cdot 6k = 12k^2$.
Finally, we define $R_{0}, R_{1},\ldots, R_b$. Let:
\begin{itemize}
	\item $\tau_0 < \tau_1 < \dots < \tau_{b - 1}$ denote the sorted elements of $\Mis(T, \QT)$,
	\item $R_i := (\tau_{i - 1} \dd \tau_i]$ for all $0 \le i \le d$, where we set $\tau_{-1} := -1$ and $\tau_b := |T|$.
\end{itemize}
This is illustrated in Figure~\ref{fig:regions}.

The elements of $\Mis(T, \QT)$ can be computed in $\Oh(k)$ time by extracting
the mismatches with longest common prefix queries, which allows us to
find the regions in $\Oh(\poly(k))$ time.
Similarly, we can compute the elements of $\Mis(P, \QP)$ in $\Oh(k)$ time, so
we can also compute the set $E$ in $\Oh(\poly(k))$ time.
We stress that the regions are the same for every considered pattern $P$
(but the set $E$ does depend on the pattern $P$).

\begin{figure}[h]
\begin{center}
\includegraphics[width=\textwidth]{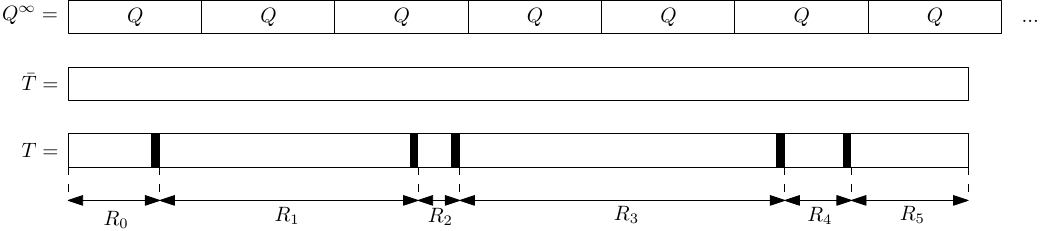}
\caption{Regions $R_{0},R_{1},\dots,R_{d}$. Black rectangles denote the elements of $\Mis(T, \QT)$.}
\label{fig:regions}
\end{center}
\end{figure}

Now we can state our extension of the structural characterization of~\cite{Charalampopoulos2020}.

\begin{restatable}{lemma}{restateKeyTheorem}
	\label{key theorem}
	There exist an integer $r$ and a set $E$, with $|E|=\Oh(k^{2})$, such that, for each pair $s,t$, with $0 \le s \le t \le b$:
	\[ 
		\Ext(P, T, A \cap R_s) \cap R_t = 
		\begin{cases}
			\left(\left(C_r \cap A \cap R_s\right) + |P|\right) \cap R_t & \text{ if } \dH(P, \QP) + t - s \le k,\\
			\text{some subset of $E + |P|$} & \text{ otherwise.}
		\end{cases}
	\]
\end{restatable}
	
	\begin{proof}
		For any $s > t$ the resulting set is trivially empty. Now let us fix any $s \le t$ and define
		\[ I_{st}^P := C_r \cap A \cap R_s \cap (R_t - |P|).\]
		We will show that either $I_{st}^P \subseteq \Occ(P, T)$ or $I_{st}^P \cap \Occ(P, T) \subseteq E$, depending on whether $\dH(P, \QP) + t - s \le k$ by using the following property:
		\begin{restatable}[]{proposition}{restateIstLemma}
		\label{Ist lemma}
			For all $x \in I_{st}^P$ we have
			$ \dH(T[x \dd x + |P|), \QT[x \dd x + |P|)) = t - s $.
		\end{restatable}
		\begin{proof}
	We know that
	\begin{itemize}
		\item $x \in R_s$, so $\tau_{s - 1} < x \le \tau_s$,
		\item $x + |P| \in R_t$, so $\tau_{t - 1} < x + |P| \le \tau_t$,
	\end{itemize}
	therefore
	\[ [x, x + |P|) \cap \set{\tau_0, \tau_1, \dots, \tau_{b - 1}} = \set{\tau_s, \tau_{s + 1}, \dots, \tau_{t - 1}},\]
	and finally
	\begin{align*}
		\dH(T[x \dd x + |P|), \QT[x \dd x + |P|)) &= 
		\big|[x, x + |P|) \cap \Mis(T, \QT)\big| \\
		&= \big|[x, x + |P|) \cap \set{\tau_0, \tau_1, \dots, \tau_{b - 1}}\big|  \\
		& = \big|\set{\tau_s, \tau_{s + 1}, \dots, \tau_{t - 1}}\big| \\
		&= t - s. \qedhere
	 \end{align*}
\end{proof}
		
		Now assume that $\dH(P, \QP) + t - s \le k$. In that case, recall that by \eqref{triangle inequality} combined with the above lemma, for every $x \in I_{st}^P$ we have
		\eq{
			\dH(P, T[x \dd x + |P|)) &\le \dH(P, \QP) + \dH(\QT[x \dd x + |P|), T[x \dd x + |P|))
			= \dH(P, \QP) + t - s \le k.
		}
		Consequently $x \in \Occ(P, T)$, and therefore $I_{st}^P \subseteq \Occ(P, T)$.
		In that case observe that
		\eq{
			\Ext(P, T, A \cap R_s) \cap R_t
			&= \left( \left(\Occ \left(P, T\right) \cap A \cap R_s \right) + |P| \right) \cap R_t \\
			&= \left(\Occ \left(P, T\right) \cap A \cap R_s \cap \left(R_t - |P| \right) \right) + |P| \\
			&= \left(\Occ \left(P, T\right) \cap C_r \cap A \cap R_s \cap \left(R_t - |P| \right) \right) + |P| \\
			&= \left(\Occ \left(P, T\right) \cap I_{st}^P \right) + |P| \\
			&= I_{st}^P + |P|,
		}
		where the third equality follows from $\Occ(P, T) \subseteq C_r$ and the fifth from $I_{st}^P \subseteq \Occ(P, T)$.
		Since $I_{st}^P + |P| = \left(\left(C_r \cap A \cap R_s\right) + |P|\right) \cap R_t$ as desired, the proof for this case is complete. 
		
		For the second case, when $\dH(P, \QP) + t - s > k$, we need to make use of the following property:
		\begin{restatable}[]{proposition}{restateTriangleEquality}
			\label{triangle equality}
			For all $x \in C_r \setminus E$ we have
			\[
				\dH(P, T[x \dd x + |P|)) = \dH(P, \QP) + \dH(\QT[x \dd x + |P|), T[x \dd x + |P|)).
			\]
		\end{restatable}
		\begin{proof}
	Consider the triangle inequality \eqref{basic triangle inequality} stated explicitly for each position $i \in [0 \dd |T|)$:
	\[
		\dH(P[i], T[i + x]) \le \dH(P[i], \QP[i]) + \dH(\QP[i], \QT[i + x]) + \dH(\QT[i + x], T[i + x]).
	\]
	From \eqref{qp=qt} we already know that $\dH(\QP[i], \QT[i + x]) = 0$, thus
	\[
		\dH(P[i], T[i + x]) \le \dH(P[i], \QP[i]) + \dH(\QT[i + x], T[i + x]).
	\]
	We will now show that the above inequality holds with equality by considering two cases.
	The proof is completed by summing the equations for all $i \in [0 \dd |T|)$.
	\begin{enumerate}[$1^\circ$]
		\item $i + x \not \in \Mis(T, \QT)$. In that case observe that by triangle inequality
			\[ \dH(P[i], \QP[i]) \le \dH(P[i], T[i + x]) + \dH(T[i + x], \QT[i + x]) + \dH(\QT[i + x], \QP[i])\]
			and since $\dH(T[i + x], \QT[i + x]) = 0$ by assumption and $\dH(\QT[i + x], \QP[i]) = 0$ by \eqref{qp=qt} we get
			\[ \dH(P[i], \QP[i]) + \dH(\QT[i + x], T[i + x]) \le \dH(P[i], T[i + x]).\]
		\item $i \not \in \Mis(P, \QP)$. In that case observe that by triangle inequality
			\[ \dH(\QT[i + x], T[i + x]) \le \dH(\QT[i + x], \QP[i]) + \dH(\QP[i], P[i]) + \dH(P[i], T[i + x]) \]
			and similarly since $\dH(\QP[i], P[i]) = 0$ and $\dH(\QT[i + x], \QP[i]) = 0$ we again get
			\[ \dH(P[i], \QP[i]) + \dH(\QT[i + x], T[i + x]) \le \dH(P[i], T[i + x]).\]
	\end{enumerate}
	It remains to show that every $i \in [0 \dd |T|)$ falls into at least one of these two cases.
	Indeed, if for some $i$ we would have $i \in \Mis(P, \QP)$ and $i + x \in \Mis(T, \QT)$, then by the definition of $E$, we would get $x \in E$, which is a contradiction.
\end{proof}
		
		By the above, combined with \Cref{Ist lemma} for every $x \in I_{st}^P \setminus E$ we have
		\eq{
			\dH(P, T[x \dd x + |P|)) &= \dH(P, \QP) + \dH(\QT[x \dd x + |P|), T[x \dd x + |P|))
			= \dH(P, \QP) + t - s > k,
		}
		therefore $x \not \in \Occ(P, T)$, which implies $I_{st}^P \cap \Occ(P, T) \subseteq E$.
		Following the same reasoning as before, up to the fourth equality, we get
		\eq{
			\Ext(P, T, A \cap R_s) \cap R_t = \left(\Occ \left(P, T\right) \cap I_{st}^P \right) + |P| \subseteq E + |P|.
		}
		as required.
	\end{proof}

We apply Lemma~\ref{key theorem} on every pattern $P_{i}$. Whenever the second case applies, we process all occurrences
of $P_{i}$ naively. We observe that by definition we have
\[
	\Ext(P, T, A \cap R_s) \cap R_t = ((\Occ(P, T) \cap E \cap A \cap R_s) + |P|) \cap R_t.
\]
Since we already have access to the previously calculated representation of $\Occ(P, T)$,
we can simply check for each element in $E$ whether it is in $\Occ(P, T) \cap A \cap R_s$ in $\Oh(k^2)$ time, as the representation of $\Occ(P, T)$ is of size $\Oh(k^2)$, so $\Oh(|E| \cdot k^2) = \Oh(k^4)$ time in total.

For the remaining patterns that fall into the first case, we still use the naive approach if $|Q| > z$ for some threshold $z$ to be chosen later.
Since $|C_{r}| = \Oh(\ell/|Q|)$, this takes $\Oh(\ell/z)$ time per pattern. Otherwise, $|Q| \leq z$.
We partition the remaining patterns into $|Q|$ groups with the same $r$.
Formally, let $\mathcal{P}_r$ denote the set of patterns $P$ with a specific value of $r$:
\[
	\bigcup_i \Ext(P_i, T, A \cap R_s) \cap R_t = \bigcup_{r = 0}^{|Q| - 1} \bigcup_{P \in \mathcal{P}_r} \left(\left(C_j \cap A \cap R_s\right) + |P|\right) \cap R_t.
\]
We calculate the result for each $r$ separately by phrasing it as a sumset of some common set of positions with the set of pattern lengths,
where the result is then truncated to $R_t$:
\eq{
	\bigcup_{P \in \mathcal{P}_j}\left(\left(C_j \cap A \cap R_s\right) + |P|\right) \cap R_t
	= \left(\left(C_j \cap A \cap R_s\right) \oplus \set{|P| \ : \ P \in \mathcal{P}_j} \right) \cap R_t.
}
This takes $\Oh(z\ell\log \ell)$ total time by \Cref{lem:fft}. Overall, the time complexity is $\Oh(d\poly(k)+\ell+d\ell/z+z\ell\log\ell)$, which by choosing
$z=\sqrt{d/\log\ell}$ becomes $\Oh(d\poly(k)+\ell\sqrt{d\log \ell})$ as promised.

\subsection{Combining the Cases}
\label{sec:combine}

After designing an algorithm for every case, we show how to combine them to obtain the claimed bounds.

\mainone*

\begin{proof}
By Lemma~\ref{lem:reduction}, to prove the theorem it is enough to show how to solve \APEmis, where $k=1$, in
$\Oh(m^{1.5}\polylog m+N)$ time. We choose $B'=\log^{2}m$ and $B=\sqrt{m}$.
For patterns of length at most $B'$, we use Theorem~\ref{thm:veryshort}.
For patterns of length at least $B'$ but at most $B$, we use Theorem~\ref{thm:short}.
Finally, for patterns of length at least $B$, we use Theorem~\ref{thm:long}.
Summing up the time complexities, we obtain
\[\Oh(m\log^{4}m+N)+\Oh(m\log^{3}m\log\log m+N)+\Oh(m+N+m^{1.5}\log^{2}m) = \Oh(m^{1.5}\polylog m+N)
\]
as required.
\end{proof}

\maintwo*

\begin{proof}
By Lemma~\ref{lem:reduction}, to prove the theorem it is enough to show how to solve \APEmis for any constant $k\geq 1$ in
$\Oh((m^{1.5}+N)\polylog m)$ time. We choose $B=\sqrt{m}$.
For patterns of length at most $B$, we use Theorem~\ref{thm:short}.
Finally, for patterns of length at least $B$, we use Theorem~\ref{thm:long}.
Summing up the time complexities, we obtain
\[\Oh(m^{1.5}\log^{k}m\log\log m+N\log^{k+1}m)+\Oh(m+N+m^{1.5}\log^{2}m) = \Oh((m^{1.5}+N)\polylog m)
\]
as required.
\end{proof}

\bibliographystyle{abbrv}
\bibliography{biblio}

\end{document}